\documentclass{cimento}

\usepackage{amsfonts}
\usepackage{amssymb}
\usepackage{dsfont}

\usepackage{xspace}
\usepackage[all]{xy}
\usepackage{amsthm}
\usepackage{amsopn}
\usepackage{footmisc}

\def\Eq#1{{\begin{equation} #1 \end{equation}}}

\def \A{\mathcal{A}}
\def \L{\mathcal{L}}
\def \S{\got S}
\def \t{\widetilde}

\def \D{\mathcal{D}}

\def \V{\mathcal{V}}
\def \Z{\mathcal{Z}}
\def \B{\mathcal{B}}

\def \Nv{\mathcal N_{\mathcal {V}}}

\def \R{\mathbb{R}}

\DeclareMathAlphabet{\Meuf}{U}{euf}{m}{n}
\def\got#1{\Meuf{#1}}
\newcommand{\LJB}{\mbox{LJB--algebra}\xspace}

\newcommand{\CA}{C^*\mbox{--algebra}\xspace}

\newtheorem{theorem}{Theorem}

\newtheorem{proposition}{Proposition}
\newtheorem{definition}{Definition}
\newtheorem{lemma}{Lemma}
\newtheorem{example}{Example}

\newcommand{\bigslant}[2]{{\left.\raisebox{.2em}{$#1$}\!\middle/\!\raisebox{-.2em}{$#2$}\right.}}

\title{Reduction of Lie--Jordan algebras: Quantum}

\author{F.~Falceto\from{ins:a}\ETC,
L.~Ferro\from{ins:b}\from{ins:c},
A.~Ibort\from{ins:b} \atque
G.~Marmo\from{ins:b}\from{ins:c}}
\instlist{\inst{ins:a} Departamento de F\'{\i}sica Te\'orica. Universidad de Zaragoza\\ Plaza San Francisco s/n, 50009 Zaragoza, Spain
  \inst{ins:b} Departamento de Matem\'aticas, Universidad Carlos III de Madrid\\ Avda. de la Universidad 30, 28911 Legan\'es, Madrid, Spain
\inst{ins:c}  Dipartimento di Scienze Fisiche, INFN--Sezione di Napoli, Universit\`a di Napoli ``Federico II'', Via Cintia Edificio 6, I--80126 Napoli, Italy
}

\PACSes{
\PACSit{03.65.Fd}{Quantum Mechanics. Algebraic methods}
\PACSit{03.65.Ca}{Quantum Mechanics. Formalism}
\PACSit{03.65.Ta}{Foundations of quantum mechanics; measurement theory}
\PACSit{02.10.Hh}{Rings and algebras}
}

\begin{document}

\maketitle

\begin{abstract} In this paper we present a theory of reduction of quantum systems in the presence of symmetries and constraints. The language used is that of Lie--Jordan Banach algebras, which are discussed in some detail together with spectrum properties and the space of states. The reduced Lie--Jordan Banach algebra is characterized together with the Dirac states on the physical algebra of observables.
\end{abstract}

\section{Introduction}  

This paper is the second part of two that jointly present a theory of reduction of Lie--Jordan algebras that can be used as an alternative to deal with symmetries and local constraints in quantum physics and quantum field theories.\\
The algebraic approach to quantum systems \cite{Ha96} has had a profound influence in both the foundations and applications of quantum physics. The background for that approach is to consider a quantum system as described by a $\CA$ $\A$ whose real part are the observables of the system, and its quantum states are normalized positive complex functionals on it. 

A geometrical approach to Quantum Mechanics \cite{Ci84},\cite{As99},\cite{Carinena:2007ws},\cite{Chruscinski:2008px} has also been developed in the last twenty years and has provided useful insight into such properties of quantum systems like integrability \cite{ClementeGallardo:2008wa}, the intrinsic nature of different measures of entanglement \cite{Grabowski:2000zk},\cite{Grabowski:2005my}, etc.

Moreover a geometrical description of dynamical systems provides a natural setting to describe symmetries and/or constraints \cite{Abr67}.  For instance, if the system carries a symplectic or Poisson structure several procedures were introduced along the years, like Marsden-Weinstein reduction, symplectic reduction, Poisson reduction, reduction of contact structures, etc. However, it was soon realized that the algebraic approach to reduction provided a convenient setting to deal with reduction of classical systems \cite{Ib97},\cite{Gr94}. 

In the standard approach to quantum mechanics, constraints are imposed on the system by selecting subspaces determined by the quantum operators corresponding to the constraints of the theory, Dirac states, and equivalence of quantum states was dealt with by using the representation theory of the corresponding group of symmetries. This paper is devoted to present a fully algebraic approach to the reduction of quantum systems.

While the first part of this contribution is focused mainly in classical mechanics, this second part is devoted to the study of the quantum case. Lie-Jordan algebras are the common algebraic language in both parts.

In the first two sections we introduce Lie--Jordan Banach algebras and discuss some useful spectral properties and the Cauchy--Schwarz inequalities. Then we will address the problem of the reduction of quantum systems as the reduction of Lie--Jordan Banach algebras. We will first consider the reduction in the presence of symmetries of the system and then in the presence of constraints. A precise characterization of the reduced algebra and an explicit description of the reduced states will be obtained. In the quantum case, the reduction procedure is more subtle than the classical one, due to the requirement imposed on the algebra to be non (Jordan) associative.
\bigskip

\section{Lie--Jordan Banach algebras}\label{second}

We recall here the definition of Lie--Jordan algebra already given in the first part of this contribution and then supplement the definition with a Banach structure. A Lie--Jordan algebra is a real vector space $\L$ equipped with the symmetric Jordan product $\circ$ and the antisymmetric Lie bracket $\left[\cdot,\cdot\right]$ satisfying 
the Jacobi identity
\begin{equation}\label{jacobi}
 \left[\,\left[\,a,b\,\right],c\,\right] + \left[\,\left[\,c,a\,\right],b\,\right] + \left[\,\left[\,b,c\,\right],a\,\right] = 0.
\end{equation}
the Leibniz's identity
\begin{equation}\label{leibniz}
 \left[a,b\circ c\right] = \left[a,b\right]\circ c + b\circ \left[a,c\right],
\end{equation}
and the associator identity
\begin{equation}\label{associator}
 (a\circ b)\circ c - a \circ (b \circ c) = \hbar^2 \left[\,b,\left[\,c,a\,\right]\,\right] ,
\end{equation}
for some $\hbar \in \R$.

Combining (\ref{leibniz}) and (\ref{associator})  
we obtain the weak associative property of Jordan algebras
\begin{equation}\label{jassociativity}
(a^2 \circ b) \circ a = a^2 \circ (b \circ a). 
\end{equation}
\begin{definition}[Lie--Jordan Banach algebra]
A Lie--Jordan algebra $\L$, complete with respect to a norm $\| \cdot \|$ that satisfies:
\begin{enumerate}
\item $\| a \circ b\| \leq \|a\|\ \|b\|$
\item $\|[a, b]\| \leq |\hbar|^{-1} \|a\|\ \|b\|$ 
\item $\| a^2 \| = \| a\|^2$
\item $\|a^2\| \leq \|a^2 + b^2\|$,
\end{enumerate}
forall$,a,b \in \L$
is called a Lie--Jordan--Banach algebra, \LJB for short.
\end{definition}
Notice that if we are given a LJB--algebra $\L$, by taking combinations of the two products we can define an associative product on the complexification $\L^\mathbb{C} = \L \oplus i \L$.  
Specifically, we define:
\begin{equation}
  x y = x\circ y - i \hbar \left[x,y\right].
\end{equation}
Such associative algebra equipped with the involution $x^* = (a + ib)^* = a - ib$ and the norm $\|x\| = \|x^*x\|^{1/2}$ is the unique $\CA$ whose real part is precisely $\L$ \cite{FFIM13}.

\bigskip

\section{Spectrum and states of Lie--Jordan Banach algebras}

\begin{definition}
 Let $\L$ be a \LJB. The spectrum $\sigma(a)$ of $a \in \L$ is defined as the set of those $z \in \R$ for which $a - z\mathds{1}$ has no inverse in $\L$. 
\end{definition}
Note that a \LJB $\L$ is a complete order unit space with respect to the positive cone 
\begin{equation}\label{positive cone}
 \L^+ = \{\,a^2 \mid a\in \L\,\}
\end{equation}
or equivalently an element is positive if its spectrum is positive.\\
We shall in this section prove some useful properties of the spectrum and then the Cauchy--Schwarz like inequalities.
\begin{lemma}
 $$\sigma(a_1^2 + a_2^2 + \lambda[a_1,a_2]) \cup \{0\} = \sigma(a_1^2 + a_2^2 - \lambda[a_1,a_2]) \cup \{0\}$$ $\forall\, a_1,a_2 \in \L$ and $\forall\, \lambda \in \R$.
\end{lemma}
\begin{proof}
 For $z \neq 0$ the invertibility of $a_1^2 + a_2^2 + \lambda[a_1,a_2] - z\mathds{1}$ implies the invertibility of $a_1^2 + a_2^2 - \lambda[a_1,a_2] - z\mathds{1}$. Namely, one computes that
\begin{eqnarray*}
 (a_1^2 + a_2^2 + \lambda[a_1,a_2] - z\mathds{1})^{-1} &=& z^{-1} \{2a_1 \circ (b \circ a_1) - a_1^2 \circ b + 2 a_2 \circ (b \circ a_2) +
\\ && -a_2^2 \circ b + 2 [a_1,b\circ a_2] + 2 a_1 \circ [b,a_2] -\mathds{1}\}
\end{eqnarray*}
with $b = \{a_1^2 + a_2^2 + \lambda[a_1,a_2] - z\mathds{1}\}^{-1}$.
\end{proof}
\begin{lemma}
 If the spectrum $\sigma(a_1^2 + a_2^2 + \lambda[a_1,a_2])$ is negative, then $a_1^2 + a_2^2 + \lambda[a_1,a_2] = 0$
$\forall\, a_1,a_2 \in \L$ and $\forall\, \lambda \in \R$.
\end{lemma}
\begin{proof}
 Note that $a_1^2 + a_2^2 - \lambda[a_1,a_2] = 2 a_1^2 + 2 a_2^2 - (a_1^2 + a_2^2 + \lambda[a_1,a_2])$ and then under the assumptions of the lemma $\sigma(a_1^2 + a_2^2 - \lambda[a_1,a_2]) \subset \R^+$. This implies, using the previous lemma, that $\sigma(a_1^2 + a_2^2 + \lambda[a_1,a_2]) = \{0\}$ and therefore it is 0.
\end{proof}
\begin{theorem}\label{posa*a}
 $$X = a_1^2 + a_2^2 - \lambda[a_1,a_2] \in \L^+$$
$\forall\, a_1,a_2 \in \L$ and $\forall\, \lambda \in \R$.
\end{theorem}
\begin{proof}
 By applying the continuous functional calculus, it is well known that every $X \in \L$ has the decomposition $X = X_+ + X_-$ \cite{La98},\cite{Conway}, where $X_+,X_- \in \L^+$ and $X_+\circ X_- = [X_+,X_-] = 0$. It follows that $X_-^3 = -(b_1^2 + b_2^2 - \lambda[b_1,b_2]) \geq 0$ with $b_1 = a_1 \circ X_- + \lambda[a_2,X_-]$ and $b_2 = \lambda[a_1,X_-]+a_2\circ X_-$. But $X_-^3 = -2 b_1^2 - 2 b_2^2 + (b_1^2 + b_2^2 + \lambda [b_1,b_2])$ which is a negative quantity and then in turn implies that $X_-=0$ and then $X = X_+ \geq 0$.
\end{proof}

The space of states $\S(\L)$ of a Jordan--Banach algebra consists of all real normalized positive linear functionals on $\L$, i.e. \begin{equation}
 \rho\colon \L \to \mathbb{R}                                                                                                                                          
\end{equation}
 such that $\rho(\mathds{1}) = 1$ and $\rho(a^2) \geq 0, \ \forall\,a \in \L$. The state space is convex and compact with respect to the $\mathrm{w}^*$--topology. We shall now prove the Lie--Jordan algebra version of the Cauchy--Schwarz inequalities.
\begin{theorem}
Let $\L$ be a $\LJB$ and $\rho$ a state on $\L$. Then if $a,b \in \L$ we have
 \begin{equation}\label{CSJ}
 \rho(a\circ b)^2 \leq \rho(a^2)\rho(b^2),
\end{equation}
and
\begin{equation}\label{CSL}
 \rho(\left[a,b\right])^2 \leq \frac{1}{\hbar^2} \rho(a^2)\rho(b^2).
\end{equation}
\end{theorem}
\begin{proof}
 Let $\lambda \in \R$, then we have
$$0 \leq \rho((\lambda a + b)^2) = \lambda^2 \rho(a^2) + 2 \lambda \rho(a \circ b) + \rho(b^2).$$
If $\rho(a^2) = 0$ then $\rho(a\circ b) = 0$ since $\lambda$ is arbitrary. If $\rho(a^2) \neq 0$, let $\lambda = - \rho(a \circ b) \rho(a^2)^{-1}$, and the first proof is immediate.\\
The second inequality is proved similarly by using the positivity of $a^2 + b^2 + 2\hbar[a,b]$ as stated in Thm. (\ref{posa*a}).
\end{proof}

\begin{example}
The self-adjoint subalgebra $\B_{sa}$ of the algebra $\B(\mathcal H)$ of bounded linear operators on a Hilbert space $\mathcal H$ with the operator norm is a Lie--Jordan Banach algebra and the states are the positive linear functional on $\B_{sa}$. Let $\varphi$ be a continuous state with respect to the ultrastrong topology on $\B(\mathcal H)$, then there is a positive linear trace class operator $\hat\rho \in \B_{sa}$ such that
\begin{equation}
 \varphi(a) = \mathrm{Tr}(\hat\rho a)
\end{equation}
for all $a \in \B_{sa}$.\\ Conversely, if $\hat\rho$ is a positive trace class operator, then the functional $a \to \mathrm{Tr}(\hat\rho a)$ defines an ultrastrongly continuous positive linear functional on $\B_{sa}$.

\end{example}

\bigskip

\section{Reduction by symmetries}

We can consider a symmetry associated to a Lie group $\mathcal G$ with Lie algebra $\mathfrak{g}$ as a map
\Eq{\hat g\colon \mathcal{G} \to \mathrm{Aut}(\L,\circ)}
which assigns to each element $g$ of the group, an automorphism $U(g)$ of the Lie--Jordan algebra algebra $\L$. Let $\xi \in \mathfrak{g}$, then the infinitesimal generator of the symmetry is defined as
\Eq{\hat\xi(a) = \left.\frac{d}{ds} U(\exp^{s\xi})(a)\right|_{s=0}}
and is a Jordan derivation, in fact $\forall\, a,b \in \L$
\Eq{\hat\xi(a \circ b) = \hat\xi(a) \circ b + a \circ \hat\xi(b)}
i.e. $\hat\xi \in \mathrm{Der}(\L,\circ)$. In addition $\hat\xi$ is a skew derivation \cite{Al98},\cite{FFIM13} since it preserves the positive cone of the algebra $\hat\xi(\L^+) \subset \L^+$ and hence, by a theorem on Jordan derivations \cite{FFIM13}, there exists $J \in \L$ such that $\forall\, a \in \L$, 
  \Eq{\hat\xi(a) = [J,a].}
It follows that $\hat\xi$ is also a Lie derivation.\\
Then if $\D$ is a set of derivations describing the symmetries of the $\LJB$, we can introduce the subspace
\Eq{\mathcal F_{\D} = \{ a \in \L \mid \hat\xi(a) = 0\quad \forall\, \xi \in \D\}}
which is easily seen to be the Lie--Jordan subalgebra representing the symmetrical physical algebra of observables.\\
The states on the reduced algebra are simply given by the restriction of the states of the full algebra on the subalgebra:
\begin{equation}
 \S(\mathcal F_{\D}) = \{\left.\rho\right|_{\mathcal F_{\D}}\ \mathrm{s.t.}\ \rho \in \S(\L)\}.
\end{equation}

\begin{example}
 A simple example is the quantum reduction of a free particle in $\R^3$ to the $S^2$ sphere. Consider the formal Lie--Jordan algebra $\L$ generated by the operators $\{\hat r,\hat{p_r},\hat\theta,\hat{p_\theta},\hat\phi,\hat{p_\phi}\}$ and the symmetries $\D = \{[\hat r,\cdot],[\hat{p_r},\cdot]\}$. Then the reduced Lie--Jordan algebra according to the above prescription is
$$\mathcal F_{\D} = \{ a \in \L \mid [\hat r,a],[\hat{p_r},a]=0\}$$
that is the algebra of observables of a particle on the sphere which is generated by $\{\hat\theta,\hat{p_\theta},\hat\phi,\hat{p_\phi}\}$.
\end{example}

\bigskip

\section{Reduction by constraints}\label{LJ red}

A quantum system with constraints is a pair $(\L,\mathcal C)$ where the field algebra $(\L,\circ,\left[\cdot,\cdot\right])$ is a unital $\LJB$ containing the constraint set 
$\mathcal C$ \cite{Grundling:1984sq},\cite{Grundling:1998zn}.
The constraints select the physical state space, also called Dirac states
\begin{equation}
 \S_D = \{\,\omega \in \S(\L) \mid \omega(c^2) = 0,\quad\forall\,c \in \mathcal C \,\}
\end{equation}
where $\S(\L)$ is the state space of $\L$.
We define the vanishing subalgebra $\V$ as:
\begin{equation}
 \V = \{\,a \in \L \mid \omega(a^2) = 0,\quad\forall\,\omega \in \S_D \,\}.
\end{equation}
\begin{proposition}
 $\V$ is a non-unital LJB--subalgebra.
\end{proposition}
\begin{proof}
 Let $a,b \in \V$. From (\ref{associator}) it follows:
\begin{equation}\label{boh}
 (a\circ b)^2 = \hbar^2\ [b,[a\circ b, a]\,] + a \circ (b \circ (a \circ b)\,).
\end{equation}
If we introduce $c = [a\circ b, a]$ and $d = b \circ (a \circ b)$, Eq. (\ref{boh}) becomes:
\begin{equation}\label{boh2}
 (a\circ b)^2 = \hbar^2\ [b,c\,] + (a \circ d\,).
\end{equation}
 From the inequalities (\ref{CSJ})(\ref{CSL}) it is easy to show that if $\omega(a^2) = 0$ then
\begin{equation}\label{boh3}
 \omega(a\circ b) = 0 = \omega([a,b]) \quad \forall\, b \in \L.
\end{equation}
Then if we apply the state $\omega$ to the expression (\ref{boh2}), from (\ref{boh3}) it follows:
\begin{equation}
 \omega(\,(a\circ b)^2\,) = \hbar^2\ \omega([b,c\,]) + \omega(a \circ d\,) = 0.
\end{equation}
By definition of $\V$, this means that $\forall\, a,b \in \V$, $a \circ b \in \V$.\\
By applying the state $\omega$ to the relation
\begin{equation}
  (a\circ b)^2 - \hbar^2\ [a,b]^2 = a \circ (b \circ (a \circ b)) - \hbar^2\ a \circ [b,[a,b]\,],
\end{equation}
we obtain $\omega(\,[a,b]^2\,) = \omega((a\circ b)^2) = 0$, that is 
$[a, b] \in \V$, $\forall\, a,b \in \V$. 
Hence $\V$ is a Lie--Jordan subalgebra and by definition of state $\omega(\mathds{1}) = 1$, hence $\mathds{1} \notin \V$.\\
$\V$ also inherits the Banach structure since it is defined as the intersection of closed subspaces.
\end{proof}

We can use the vanishing subalgebra to give an alternative useful description of 
the Dirac states.
\begin{proposition}\label{newDirac} With the previous definitions we have
\begin{equation} 
\S_D=\{\omega\in\S(\L) \mid \omega(a)=0,\quad \forall\, a\in\V\}
\end{equation}
\end{proposition}
\begin{proof}
As $\V$ is a subalgebra and it contains $\mathcal C$ it is clear that
the right hand side is included into $\S_D$. 

To see the other inclusion it is enough to consider that 
for any state $\omega(a)^2\leq\omega(a^2)$, therefore any Dirac state
should vanish  on $\V$.
\end{proof}

Define now the Lie normalizer as
\begin{equation}
\mathcal \Nv = \{\,a \in \L \mid \left[a,\V\right] \subset \V\,\}
\end{equation}
which corresponds roughly to Dirac's concept of ``first class variables'' \cite{Dirac}.

\begin{proposition}
 $\Nv$ is a unital \LJB and $\V$ is a Lie--Jordan ideal of $\Nv$.
\end{proposition}
\begin{proof}
 Let $a,b \in \Nv$ and $v \in \V$. Then by definition of normalizer it immediately follows:
$$[[a,b],v] = [[a,v],b] + [[v,b],a] \in \V.$$
Let us now prove that $\forall\, v \in \V,\ v \circ a \in \V$, this is $\V$ is a Jordan ideal of $\Nv$:
$$\omega((v\circ a)^2) = \hbar^2\ \omega([a,[v\circ a,v]]) + \omega(v\circ(a\circ(a\circ v)))$$
which gives zero by repeated use of properties (\ref{CSJ}) and (\ref{CSL}).\\
Then it becomes easy to prove that $\Nv$ is a Jordan subalgebra:
$$[a\circ b, v] = [a,v] \circ b + a \circ [b,v] \in \V.$$
Finally, since the Lie bracket is continuous with respect to the Banach structure, it also follows that $\Nv$ inherits the Banach structure by completeness.
\end{proof}

In the spirit of Dirac, the physical algebra of observables in the presence of the constraint set $\mathcal C$ is represented by the \LJB $\Nv$ which can be reduced by the closed Lie--Jordan ideal $\V$ which induces a canonical Lie--Jordan algebra structure in the quotient:
\begin{equation}
\t\L = \bigslant{\Nv}{\V}.
\end{equation}
We will denote in the following the elements of $\t\L$ by $\t a$.\\
The quotient Lie-Jordan algebra $\t\L$ carries the quotient norm, 
\begin{equation}
 \|\t a\| = \| \left[a\right]\| = \displaystyle\inf\limits_{b \in \V} \|a +b \|  ,
\end{equation}
where $a \in \Nv$ is an element of the equivalence class $\left[a\right]$ of $\Nv$ with respect to the ideal $\V$.
The quotient norm provides a \LJB structure to $\t\L$.

\begin{example}
 Consider again the quantum reduction of the free particle constrained with a fixed value of the angular momentum. The constraint operator is expressed by $c = \hat L^2 - \hbar^2 l(l+1)\mathds{1}$. The set $\mathcal H_l$ of Dirac states is given by the convex span of the states $\{\, |m,l\rangle\langle m,l|,\ m = 0,\pm 1, \ldots, \pm l\,\}$ and the vanishing subalgebra $\V$ consists of the operators vanishing on $\mathcal H_l$: 
$$ \V = \{\, \hat a \in \L \mid \left.\hat a\right|_{\mathcal H_l} = 0\,\}.$$ It follows that the Lie normalizer $\Nv$ is given by the operators preserving $\mathcal H_l$:
$$\Nv = \{\, \hat a \in \L \mid \hat a (\mathcal H_l) \subset \mathcal H_l\,\},$$
and the reduced algebra $\t\L = \Nv / \V = \{\, \t{\hat a} \colon \mathcal H_l \to \mathcal H_l \,\}$, which gives rise to a reduced nonlinear dynamics out of an unreduced linear one.
\end{example}

\subsection{The space of states of the reduced \LJB}\label{reduced_states}

The purpose of this last section is to discuss the structure of the space of states of the reduced \LJB with respect to the space of states of the unreduced one.

Let $\L$ be a $\LJB$ and $\V$ its vanishing subalgebra with respect to a constraint set $\mathcal C$ and $\mathcal{N}_{\V}$ the Lie normalizer of $\V$. Then we will denote as before by $\t{\L}$ the reduced Lie--Jordan Banach algebra $\Nv/ \V$ and its elements by $\t a$.

Let $\t{\S} = \S(\t{\L})$ be the state space of the reduced \LJB $\t{\L}$, i.e. $\t{\omega} \in \t{\S}$ means that $\t{\omega}(\t a^2) \geq 0 \ \forall\,\t a \in \t{\L}$, and $\t{\omega}$ is normalized. Notice that if $\L$ is unital, then $\mathds{1} \in \Nv$ and $\mathds{1} + \V$ is the unit element of $\t{\L}$. We will denote it by $\t{\mathcal{\mathds{1}}}$.\\ We have the following:

\begin{lemma}\label{lemma1}
 There is a one-to-one correspondence between normalized positive linear functionals on $\t{\L}$ and normalized positive linear functionals on $\Nv$ vanishing on $\V$.
\end{lemma}

\begin{proof}
 Let $\omega'\colon \Nv\to \mathbb{R}$ be positive. 
The positive cone on $\t{\L}$ consists of elements of the form $\t a^2 = (a + \V)^2 = a^2 + \V$, i.e.
\begin{equation}
 \mathcal K^+_{\t{\L}} = \{\,a^2 + \V \mid a \in \Nv\,\} = \mathcal K^+_{\Nv} + \V.
\end{equation}
Thus if $\omega'$ is positive on $\Nv$, $\omega'(a^2)\geq 0$, hence:
\begin{equation}
 \omega'(a^2 + \V) = \omega'(a^2) + \omega'(\V)
\end{equation}
 and if $\omega'$ vanishes on the closed ideal $\V$, then $\omega'$ induces a positive linear functional on $\t{\L}$.  Clearly $\omega'$ is normalized then the induced functional is normalized too because $\t{\mathcal{\mathds{1}}} = \mathds{1} + \V$. 

Conversely, if $\t{\omega}\colon\t{\L} \to \mathbb{R}$ is positive and we define
\begin{equation}
 \omega'(a) = \t{\omega}(a + \V)
\end{equation}
then $\omega'$ is well-defined, positive, normalized and $\left.\omega'\right|_{\V}=0$.
\qedhere
\end{proof}

Notice also that given a positive linear functional on $\Nv$ there exists an extension of it to $\L$ which is positive too.

\begin{lemma}\label{lemma2}
 Given a closed Jordan subalgebra $\Z$ of a \LJB $\L$ such that $\mathds{1} \in \Z$ and $\omega'$ is a normalized positive linear functional on $\Z$, then there exists $\omega\in\S({\L})$ such that $\omega(a) = \omega'(a), \ \forall\,a \in \Z$.
\end{lemma}

\begin{proof}
  Since $\L$ is a JB--algebra, it is also a Banach space. Due to the Hahn--Banach extension theorem \cite{Conway}, there exists a continuous extension $\omega$ of $\omega'$, i.e. $\omega(a) = \omega'(a), \ \forall\,a \in \Z$, and moreover $\|\omega\|=\|\omega'\|$.

From the equality of norms and the fact that $\omega'$ is positive we have
$\|\omega\|=\omega'(\mathds{1})$,
but $\omega$ is an extension of $\omega'$ then $\|\omega\|=\omega(\mathds{1})$, 
which implies that $\omega$ is a positive functional and satisfies all the requirements
stated in the lemma.
\qedhere
\end{proof}

We can now prove the following:
\begin{theorem}
 The set $\S_D(\Nv)$ of Dirac states on $\L$ restricted to $\Nv$ is in one-to-one correspondence with the space of states of the reduced \LJB $\t\L$.
\end{theorem}
\begin{proof}
In Prop. \ref{newDirac} we characterised the Dirac states as those that vanish
on $\V$. Combining this result with that of Lemma \ref{lemma1} the proof
follows. 
\qedhere
\end{proof}

\acknowledgments
 This work was partially supported by MEC grants FPA--2009-09638,
MTM2010-21186-C02-02, QUITEMAD programme and DGA-E24/2.  
G.M. would like to acknowledge the support provided by the Santander/UCIIIM Chair of Excellence programme 2011-2012. 
\vskip 1cm

\end{document}